\documentclass{amsart}

\usepackage{graphicx}
\usepackage{color}
\usepackage{amsmath}
\usepackage{shuffle}
\usepackage{mathtools}
\usepackage{tikz}
\usepackage{subfigure}

\newcommand\independent{\protect\mathpalette{\protect\independenT}{\perp}}
\def\independenT#1#2{\mathrel{\rlap{$#1#2$}\mkern2mu{#1#2}}}

\newtheorem{theorem}{Theorem}[section]
\newtheorem{proposition}[theorem]{Proposition}
\newtheorem{lemma}[theorem]{Lemma}

\theoremstyle{definition}
\newtheorem{definition}[theorem]{Definition}

\theoremstyle{remark}

\numberwithin{equation}{section}

\newcommand{\dd}{\mathrm{d}}
\newcommand{\ev}{\mathrm{ev}}
\newcommand{\R}{\mathbb{R}}
\newcommand{\N}{\mathbb{N}}
\newcommand{\Z}{\mathbb{Z}}

\newcommand{\cP}{\mathcal{P}}
\newcommand{\oT}{\overline{T}}
\newcommand{\llangle}{\langle\langle}
\newcommand{\rrangle}{\rangle\rangle}

\begin{document}

\title{Iterated integrals and population time series analysis}

\author{Chad Giusti}
\address{Department of Mathematical Sciences, University of Delaware, Newark, DE 19716}
\email{cgiusti@udel.edu}

\author{Darrick Lee}
\address{Department of Mathematics, University of Pennsylvania, Philadelphia, PA 19104}
\email{ldarrick@sas.upenn.edu}

\date{\today}

\begin{abstract}
One of the core advantages topological methods for data analysis provide is that the language of (co)chains can be mapped onto the semantics of the data, providing a natural avenue for human understanding of the results. Here, we describe such a semantic structure on Chen's classical iterated integral cochain model for paths in Euclidean space. Specifically, in the context of population time series data, we observe that iterated integrals provide a model-free measure of pairwise influence that can be used for causality inference. Along the way, we survey recent results and applications, review the current standard methods for causality inference, and briefly provide our outlook on generalizations to go beyond time series data.
\end{abstract}

\maketitle

The growing availability of population time series data drawn from observations of complex systems is driving a concomitant demand for analytic tools. Of particular interest are methods for extracting features of the time series which provide human-understandable links between the observed function and the unknown structure or organizing principles of the system. 

Over the last decade, substantial work has been done using persistent homology for time series analysis, including~\cite{cohen-steiner_vines_2006,munch_probabilistic_2015,perea_sliding_2015}. However, there are still substantial mathematical and conceptual barriers to direct interpretation of persistence diagrams in terms of the underlying data; most successes have come from statistical analyses of families of diagrams, which provide some measure of discriminatory power between systems. Thus, it is common to rely on persistence for classification. However, the success of such a program must be measured against the capabilities of modern machine learning tools, which appear capable of being tuned to out-perform topological methods. Using the results of topological computations as a pre-processing step for machine learning tools has been successful, providing a rich-but-low-dimensional feature set which retains strong discriminatory power, but human interpretation of the results suffers from the same difficulties as before.

It is the authors' opinion that one of applied topology's greatest potential advantages is the ability to ask specific, fundamentally qualitative questions of data sets and compute answers in a context and language that humans can interpret. The machinery of (co)homology provides a blueprint for asking and answering such questions, in the form of (co)chain models. However, rather than encoding data and then searching for meaning in the (co)homology, the authors propose selecting or designing the encoding topological space explicitly for the purpose of leveraging a (co)chain model which naturally encodes questions and answers of interest. 

This is a nuanced undertaking, perhaps best undertaken in the context of a collaboration between mathematicians and scientific domain experts. However, in the case of certain general data types, we can rely on the substantial extant literature on cochain models in algebraic topology for inspiration. For example, in the case of our motivating question about families of time series, we can make use of the \emph{iterated integral} model for cochains on $P\mathbb{R}^N$, originally developed by K.~T.~Chen~\cite{chen_iterated_1954,chen_integration_1957,chen_integration_1958,chen_iterated_1977}, more recently adapted to the study of stochastic differential equations~\cite{friz_multidimensional_2010,lyons_differential_2007}, and finally picked up by the machine learning community in the guise of \emph{path signatures} as a feature set for paths. In this paper, we will survey path signatures, the $0$-cochains in Chen's iterated integral model, and their fundamental properties, discuss how they have been applied to characterize \emph{cyclic} structure in observed time series, and offer a new interpretation of lower-order iterated integrals as a measure of causality among simultaneously observed time series. Finally, we briefly provide our outlook on how higher cochains, and cochain models of more general mapping spaces may be leveraged for data analysis beyond time series.

\section{Path signatures as iterated integrals}
\label{sec:iterated}


Consider a collection of $N$ simultaneous real-valued time series, $\gamma_i : [0, 1] \to \mathbb{R}$, $i = 1, \dots, N$, thought of as coordinate functions for a path $\Gamma \in P\mathbb{R}^N = C([0, 1], \mathbb{R}^N)$. Foundational work by K.T.~Chen used \textit{iterated integrals} to produce a rational cochain model for this space.

\begin{definition}
\label{def:ii}
Suppose $\dd x_1, \ldots, \dd x_N$ are the standard 1-forms for $\R^N$. For $t\in [0,1]$, let $\Gamma_t = \Gamma|_{[0,t]}$. For $i \in [N]$, define a path
\begin{equation*}
    S^i(\Gamma)(t) = \int_{\Gamma_t} \dd x_i = \int_0^t \Gamma^* \dd x_i(s) = \int_0^t \dd \gamma_i(s).
\end{equation*}
Let $I = (i_1, \ldots, i_k)$, where $i_l \in [N]$. Higher order paths are inductively defined as
\begin{equation*}
    S^I(\Gamma)(t) = \int_0^t S^{(i_1, \ldots, i_{k-1})}(\Gamma)(s) \dd \gamma_{i_k}(s).
\end{equation*}
The \textit{iterated integral} of $\Gamma$ with respect to $I$ is defined to be $S^I(\Gamma) \coloneqq S^I(\Gamma)(1)$.
\end{definition}

We can also define the iterated integral in a non-inductive way. Let $\Delta^k$ be the simplex
\begin{align*}
    \Delta^k = \left\{(t_1, \ldots, t_k)\, | \, 0 \leq t_1 \leq \ldots \leq t_k \leq 1\right\}.
\end{align*}
By direct computation, we have $\Gamma^* \dd x_i = \gamma'_i(t) \dd t$.
Then, the iterated integral of $\Gamma$ with respect to $I$ is equivalently defined as
\begin{equation}
\label{eq:ii_function}
    S^I(\Gamma) = \int_{\Delta^k} \gamma'_{i_1}(t_1) \gamma'_{i_2}(t_2) \ldots \gamma'_{i_k}(t_k) \, \dd t_1 \dd t_2 \ldots \dd t_k.
\end{equation}

These iterated integrals with respect to a fixed $I$ can be viewed as functions $S^I: P\R^N \rightarrow \R$ on $P\R^N$. Chen generalized this concept of iterated integration to produce \textit{forms} on $P\R^N$, which fit together to generate a cochain model of $P\R^N$. The iterated integrals defined here are the $0$-cochains of this cochain model. A summary of this construction is included in Appendix~\ref{sec:pathspace}, and a brief discussion of higher cochains is in Section~\ref{sec:outlook} .\medskip



In this section, we discuss various properties and characterizations of these iterated integrals, in preparation for their application to time series analysis in the following section. A wide class of paths in which these theorems hold is the class of bounded variation. For the remainder of the paper, we consider $\R^N$ equipped with the standard Euclidean norm, denoted $\|\cdot\|$. 

\begin{definition}
    Let $\Gamma \in P\R^N$. The \textit{1-variation} of $\Gamma$ on $[0,1]$ is defined as 
    \begin{equation}
    \label{eq:1var}
        |\Gamma|_{1-var} \coloneqq \sup_{(t_i) \in \cP([0,1])} \sum_{i} \| \Gamma(t_i) - \Gamma(t_{i-1}) \| ,
    \end{equation}
    where $\cP([0,1])$ is the set of all finite partitions of $[0,1]$. Paths in the class
    \begin{equation*}
        BV(\R^N) = \left\{ \Gamma \in P\R^N \; | \; |\Gamma|_{1-var} < \infty\right\}
    \end{equation*}
    are the paths of \textit{bounded variation} on $[0,1]$. Note that the $1$-variation is a norm on $BV(\R^N)$.
\end{definition}

 The collection of iterated integrals of $\Gamma$ with respect to all multi-indices $I$ is called the \textit{path signature} of $\Gamma$, denoted $S(\Gamma)$. The path signature can be represented as an element of the formal power series algebra of tensors (or also viewed as non-commutative indeterminates $X = \{X_1, \ldots, X_N\}$, denoted $\oT(\R^N)$,
\begin{equation*}
    S(\Gamma) = 1+ \sum_{k=1}^\infty \sum_{I = (i_1, \ldots, i_k)} S^I(\Gamma) X_{i_1} \otimes \ldots \otimes X_{i_k}.
\end{equation*}

Several  of  the  basic  properties  of  these  path  signatures  provide  evidence  that they are potentially useful for time series analysis.

\begin{proposition}
    Suppose $\Gamma \in BV(\R^N)$, $\phi:[0,1] \to [0,1]$ a strictly increasing function, $a\in \R^N$, and $\lambda \in \R$.The path signature is invariant under translation,
    \begin{align*}
        S(\Gamma + a) = S(\Gamma),
    \end{align*}
    and reparametrization,
    \begin{align*}
        S(\Gamma \circ \phi) = S(\Gamma).
    \end{align*}
    Additionally, under scaling, we have
    \begin{align*}
        S(\lambda\Gamma) = 1+ \sum_{k=1}^\infty \sum_{I = (i_1, \ldots, i_k)} \lambda^k  S^I(\Gamma) X_{i_1} \otimes \ldots \otimes X_{i_k}.
    \end{align*}
\end{proposition}
\begin{proof}
    All three properties are straightforward to show using the definition of path signatures. Translation invariance is due to the translation invariance of the standard $1$-forms on $\R^N$. reparametrization invariance of the first level is given by
    \begin{align*}
        S^i(\Gamma\circ\phi)  =\int_0^1 (\gamma_i(\phi(t))' dt 
         = \int_0^1 \gamma_i'(\phi(t)) \phi'(t) dt 
         = \int_0^1 \gamma_i'(\tau) d\tau
         = S^i(\Gamma).
    \end{align*}
    Invariance for higher level signatures is shown by induction. Finally, the scaling property is clear from the definition of Equation~\ref{eq:ii_function}.
\end{proof}

Note that signatures can be defined for paths with an arbitrary closed interval $[a,b] \subset \R$ as a domain. However, without loss of generality due to reparametrization invariance, we only consider paths defined on $[0,1].$

These path signatures characterize classes of paths in $\R^N$ up to a \textit{tree-like equivalence}, originally defined in~\cite{hambly_uniqueness_2010}. In order to define the relation, we first consider concatenation of paths. Suppose $\Gamma_1, \Gamma_2 \in BV(\R^N)$, then define the concatenation of the two paths, $\Gamma_1 *\Gamma_2 \in BV(\R^N)$ by
\begin{equation}
    \Gamma_1 * \Gamma_2 (t) = \left\{
        \begin{array}{cl}
            \Gamma_1(2t) & : t\in [0,\frac12) \\
            (\Gamma_1(1) -\Gamma_2(0)) + \Gamma_2(2t-1) & : t \in [\frac12,1]
        \end{array}
    \right.
\end{equation}
The inverse of a path $\Gamma$ is defined to be the same path but running in the opposite direction, namely
\begin{align*}
    \Gamma^{-1}(t) = \Gamma(1-t).
\end{align*}

\begin{definition}[\cite{hambly_uniqueness_2010}]
    A path $\Gamma \in BV(\R^N)$ is a \textit{tree-like path} in $\R^N$ if there exists some positive real-valued continuous function $h$ defined on $[0,1]$ such that $h(0) = h(1) = 0$ and such that
    \begin{equation}
    \label{eq:treelike}
        \|\Gamma(t) - \Gamma(s)\| \leq h(s) + h(t) - 2 \inf_{u \in [s,t]} h(u),
    \end{equation}
    where $\| \cdot \|$ is the Euclidean norm on $\R^N$. The function $h$ is called a \textit{height function} for $\Gamma$ and if $h$ is of bounded variation, then $\Gamma$ is a \textit{Lipschitz tree-like path}.
\end{definition}

\begin{definition}
    Two paths $\Gamma_1, \Gamma_2 \in BV(\R^N)$ are \textit{tree-like equivalent}, $\Gamma_1 \sim \Gamma_2$, if $\Gamma_1 * \Gamma_2^{-1}$ is a Lipschitz tree-like path.
\end{definition}

It is shown in~\cite{hambly_uniqueness_2010} that tree-like equivalence is an equivalence relation in $BV(\R^N)$ and that concatenation of paths respects $\sim$. By defining the inverse of a path $\Gamma$ by $\Gamma^{-1}(t) = \Gamma(1-t)$, the equivalence classes $\Sigma = BV(\R^N)/\sim$ form a group under concatenation.  

The more abstract notion of a tree-like path is required when working with general bounded variation paths, but if we restrict ourselves to piecewise regular paths, we can use a much more intuitive characterization based on reductions. Specifically, a path $\Gamma$ is called \textit{reducible} if there exist paths $\alpha$, $\beta$, and $\gamma$ such that $\Gamma = \alpha * \gamma*\gamma^{-1} * \beta$ up to reparametrization, and called \textit{irreducible} otherwise. Furthermore, $\alpha * \beta$ is called a reduction of $\Gamma$.

A path $\Gamma \in P\R^N$ is \textit{regular} if $\Gamma'(t)$ is continuous and nonvanishing for all $[0,1]$. Chen~\cite{chen_integration_1958} showed that for any piecewise regular path $\Gamma$, we can obtain a unique (up to reparametrization) irreducible path by applying a finite number of reductions. Then, we can prove this simplified characterization.

\begin{lemma}
    Suppose $\Gamma\in P\R^N$ is a piecewise regular path. Then $\Gamma$ is a Lipschitz tree-like path if and only if its irreducible reduction is the constant path.
\end{lemma}
\begin{proof}

    First, suppose $\Gamma$ can be reduced to a point. Thus, $\Gamma$ can be constructed iteratively with a finite set of paths $\gamma_1, \ldots, \gamma_k$ as follows. Begin with $\Gamma_1 = \gamma_1 * \gamma_1^{-1}$, then $\Gamma_2 = \alpha_1 * \gamma_2 * \gamma_2^{-1}*\beta_1$, where $\Gamma_1 = \alpha_1 * \beta_1$. Continue in this manner until $\Gamma = \Gamma_k = \alpha_{k-1} * \gamma_k * \gamma_k^{-1} * \beta_{k-1}$. For example, consider the following point reducible path $\Gamma$ which can be built with two paths.
    
    \begin{figure}[!htbp]
    \centering
	    \includegraphics[width=0.6\textwidth]{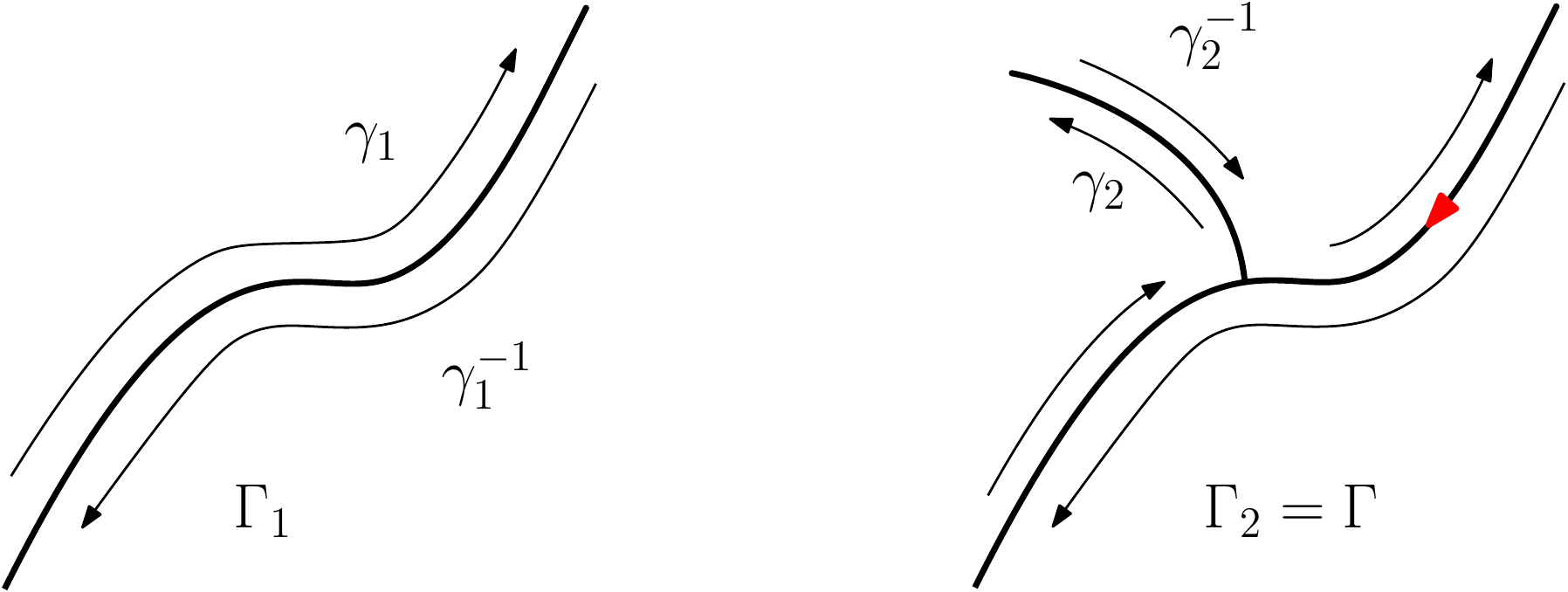}
    \end{figure}
    
    Note that $\Gamma$ will traverse each of $\gamma_1, \gamma_1^{-1}, \ldots, \gamma_k, \gamma_k^{-1}$ exactly once. Now, define $\Gamma^t$ to be the image of $\Gamma|_{[0,t]}$, and treat each of the $\gamma_i$ as the images. Then, define the height function to be
    \begin{equation*}
        h(t) = \sum_{i=1}^k \ell(\gamma_i \cap \Gamma^t) - \sum_{i=1}^k \ell(\gamma_i^{-1} \cap \Gamma^t),
    \end{equation*}
    where $\ell(\cdot)$ represents the length of the given segment. Intuitively, $h(t)$ is the length of the curve up to $\Gamma(t)$, where we subtract off any segment that has been retraced. In our example, suppose that the red arrow represents the point $\Gamma(t_0)$, which has begun to traverse along $\gamma_1^{-1}$. The corresponding height function at the point is the difference of path lengths
    
    \begin{figure}[!htbp]
    \centering
	    \includegraphics[height=11.5pt]{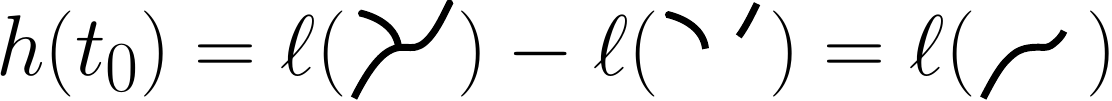}.
    \end{figure}
    
    At the end of the curve, all paths and inverse paths will have been traced so $h(1) = 0$. Note that $h(t_1) + h(t_2) - 2 \inf_{u \in [t_1,t_2]} h(u)$ represents the length of a curve from $\Gamma(t_1)$ to $\Gamma(t_2)$ which must be larger than $\|\Gamma(t_2) - \Gamma(t_1)\|$ since this is the straight line path. Lastly, the derivative of $\Gamma(t)$ is bounded over the closed interval, the arc length function and thus the height function is Lipschitz. Thus, $\Gamma$ is Lipschitz tree-like.
    
    Next, suppose $\Gamma$ is Lipschitz tree-like, and suppose to the contrary that $\Gamma$ cannot be reduced to a point. Let $\Gamma_r$ be the irreducible reduction of $\Gamma$. Then, $\Gamma * \Gamma_r^{-1}$ is point reducible, and thus Lipschitz tree-like by the first part of the proof. Thus, $\Gamma$ and $\Gamma_r$ are tree-like equivalent, so by the equivalence relation, if $\Gamma_r$ is not tree-like, then $\Gamma$ is also not tree-like. Thus, we assume $\Gamma$ is reduced so that it is irreducible.
    
    The height function $h(t)$ is Lipschitz continuous, so there exists some local maximum at $t = t_m$. Next, choose $t_1 < t_m$ and $t_2 > t_m$ such that the following hold:
    \begin{itemize}
        \item $h(t_1) = h(t_2) = h(t_m) - \epsilon$ for some $\epsilon > 0$,
        \item $\inf_{u \in [t_1,t_2]} h(u) = h(t_m) - \epsilon$, and
        \item $\Gamma(t_1) \neq \Gamma(t_2)$.
    \end{itemize}
    The first two conditions are possible because $h(t)$ is continuous, and the last condition is possible because $\Gamma$ is irreducible. Therefore, we have
    \begin{align*}
        \| \Gamma(t_2) - \Gamma(t_1)\| \leq h(t_2) + h(t_1) - 2 \inf_{u \in [t_1, t_2]} h(u) = 0,
    \end{align*}
    a contradiction. 
\end{proof}

Now we state the characterization theorem, which was proved by Chen~\cite{chen_integration_1958} for irreducible piecewise regular continuous paths, and generalized in~\cite{hambly_uniqueness_2010} to bounded variation paths $BV(\R^N) \subset P\R^N$.

\begin{theorem}[\cite{hambly_uniqueness_2010}] Suppose $\Gamma_1, \Gamma_2 \in BV(\R^N)$. Then $S(\Gamma_1) = S(\Gamma_2)$ if and only if they are tree-like equivalent.
\end{theorem}

In fact, this statement is even stronger when we consider the algebraic structure of the group of equivalence classes $\Sigma$ and the group-like elements in formal power series. An element $P \in \oT(\R^N)$ has a multiplicative inverse if and only if it has a nonzero constant term. Therefore, the restriction $\widetilde{T}(\R^N)$ to formal power series with constant term $1$ is a group under multiplication. Note that $S(\Gamma) \in \widetilde{T}(\R^N)$ by definition. One of Chen's original results~\cite{chen_iterated_1954} showed that the path signature map respects the multiplicative structure of paths and the formal power series. Namely, given $\Gamma_1, \Gamma_2 \in P\R^N$, we have
\begin{equation*}
    S(\Gamma_1 * \Gamma_2) = S(\Gamma_1) \otimes S(\Gamma_2).
\end{equation*}
Thus, the above theorem can be succinctly restated.
\begin{theorem}[\cite{hambly_uniqueness_2010}]
\label{thm:injective}
    The signature map $S: \Sigma \to \widetilde{T}(\R^N)$ is an injective group homomorphism.
\end{theorem}

That is, the path signature provides a complete set of invariants for paths up to tree-like equivalence, meaning any reparametrization-invariant property of such equivalence classes can be derived using the signature terms. Thus, any property of time series that does not rely on the parameterization can be extracted from the signature.

This point of view is further emphasized in recent results by Chevyrev and Oberhauser~\cite{chevyrev_signature_2018}, which state that a normalized variant of the signature map $\tilde{S}$ is \textit{universal} to the class $C_b(\Sigma,\R)$ of continuous bounded functions on $\Sigma$, with respect to the strict topology and is \textit{characteristic} to the space of finite regular Borel measures on $\Sigma$. Loosely speaking, universal to $C_b(\Sigma,\R)$ means that any continuous, bounded function $\phi: \Sigma \rightarrow \R$ can be approximated by a linear functional $\phi \approx \langle \ell, \tilde{S}(\cdot)\rangle$, where $\ell \in \widetilde{T}(\R^N)^*$. Namely, in the context of classification tasks, any decision boundary defined by a function in $C_b(\Sigma,\R)$ can be represented as a linear decision boundary in $\widetilde{T}(\R^N)$ under the signature map. This provides theoretical justification for the classification tasks discussed in the next section. Characteristic means that finite, regular Borel measures on $\Sigma$ are characterized by their expected normalized signatures (in the same way that probability measures with compact support on $\R^N$ are characterized by their moments).

\medskip

In addition to the multiplicative property of the signature, there exist a host of other properties, stemming from another early result of Chen that 
\begin{equation*}
    \log(S(\Gamma)) \coloneqq \sum_{j \geq 1} \frac{(-1)^{j-1}}{j} (S(\Gamma)-1)^j.
\end{equation*}
is a Lie series for any path $\Gamma$~\cite{chen_integration_1957}. This fact is equivalent to a shuffle product identity~\cite{reutenauer_free_1993}, providing an internal multiplicative structure for the path signature.

\begin{definition}
    Let $k$ and $l$ be non-negative integers. A \textit{$(k,l)$-shuffle} is a permutation of $\sigma$ of the set $\{1, 2, \ldots, k+l\}$ such that
    \begin{align*}
        \sigma^{-1}(1) < \sigma^{-1}(2) < \ldots < \sigma^{-1}(k)
    \end{align*}
    and
    \begin{align*}
        \sigma^{-1}(k+1) < \sigma^{-1}(k+2) < \ldots < \sigma^{-1}(k+1).
    \end{align*}
    We denote by $Sh(k,l)$ the set of $(k,l)$-shuffles.

Given two finite ordered multi-indices $I = (i_1, \ldots, i_k)$ and $J = (j_1, \ldots, j_l)$ , let $R = (r_1, \ldots, r_k, r_{k+1}, \ldots r_{k+1}) = (i_1, \ldots, i_k, j_1, \ldots, j_l)$ be the concatenated multi-index. The \textit{shuffle product} of $I$ and $J$ is defined to be the multiset
\begin{align*}
    I \shuffle J = \left\{ \left( r_{\sigma(1)}, \ldots r_{\sigma(k+l)}\right) \, | \, \sigma \in Sh(k,l)\right\}.
\end{align*}
\end{definition}

As an example, suppose $I = (1, 2)$ and $J = (2,3)$. Then
\begin{align*}
    I \shuffle J = \left\{ (1,2,2,3), (1,2,2,3), (2,1,2,3), (1,2,3,2), (2,1,3,2), (2,3,1,2) \right\}.
\end{align*}

\begin{theorem}[\cite{reutenauer_free_1993}]
\label{thm:shuffle}
    Let $I$ and $J$ be multi-indices in $[N]$. Then
    \begin{equation*}
        S^I(\Gamma)S^J(\Gamma) = \sum_{K \in I \shuffle J} S^K(\Gamma).
    \end{equation*}
\end{theorem}

\begin{proof}
    Let $R = (r_1, \ldots, r_k, r_{k+1}, \ldots r_{k+1}) = (i_1, \ldots, i_k, j_1, \ldots, j_l)$. Writing out the signature on the left side of the equation using Equation~\ref{eq:ii_function}, we get 
    \begin{align*}
        \int_{\Delta^k \times \Delta^l}\gamma'_{r_1}(t_1) \ldots, \gamma'_{r_{k+l}}(t_{k+l}) \, \dd t_1 \ldots \dd t_{k+l},
    \end{align*}
    and the sum on the right side is
    \begin{align*}
        \sum_{\sigma \in Sh(k,l)} \int_{\Delta^{k+l}} \gamma'_{\sigma(r_1)}(t_1) \ldots \gamma'_{\sigma(r_{k+l})} (t_{k+l}) \dd t_1 \ldots \dd t_{k+l}.
    \end{align*}

    The equivalence of the two formulas is given by the standard decomposition of $\Delta^k \times \Delta^l$ into $(k+l)$-simplices,
    \begin{align*}
        \Delta^k \times \Delta^l &= \left \{ (t_1, \ldots, t_{k+l}) \, | \, 0 < t_1 < \ldots < t_k < 1, \, 0 < t_{k+1} < \ldots < t_{k+l} < 1\right\}\\
        & = \bigsqcup_{\sigma \in Sh(k,l)} \left\{ (t_{\sigma(1)}, \ldots, t_{\sigma(k+l)}) \, | \, 0 < t_1 < \ldots < t_{k+l} < 1\right\}.
    \end{align*}
\end{proof}

Note in particular that this implies the signature terms are not independent. For example, the shuffle formula says that $S^{2,1}(\Gamma) = S^1(\Gamma)S^2(\Gamma) - S^{1,2}(\Gamma)$. Thus computation of all signature terms, even truncated to a finite level, results in redundant information. Basis sets for Lie series exist~\cite{reutenauer_free_1993}, and the set of Lyndon bases have been considered for signature computations~\cite{reizenstein_calculation_2017,reizenstein_iisignature_2018}. Further pertinent results related to Lie series can be found in~\cite{reutenauer_free_1993}.

\medskip

 Another property of central importance in data analysis is continuity of the signature map. Let $k \in \N$, then define the map $\pi_k : \oT(\R^N) \rightarrow T^k(\R^N)$ to be the projection to the $k^{th}$ tensor level. Additionally, we equip $T^k(\R^N)$ with the norm
\begin{equation*}
    |P|_{k} \coloneqq \sqrt{ \sum_{i_1, \ldots, i_k} |P^{i_1, \ldots, i_k}|^2}, \quad \textrm{for all}\quad P = \sum_{i_1, \ldots, i_k} P^{i_1, \ldots, i_k} X_{i_1} \otimes X_{i_k}.
\end{equation*}
Recall that $BV(\R^N)$ is equipped with the $1$-variation norm defined in Equation~\ref{eq:1var}. With respect to these two norms, we obtain the following continuity result.


\begin{proposition}[\cite{friz_multidimensional_2010}]
    Suppose $\Gamma_1, \Gamma_2 \in BV(\R^N)$ and $L \geq \max_{i=1,2} |\Gamma_i|_{1-var}$. Then, for all $k \geq 1$, there exist constants $C_k > 0$ such that
    \begin{equation*}
        \left| \pi_k\left( S(\Gamma_1) - S(\Gamma_2)\right) \right|_k \leq C_k L^{k-1} |\Gamma_1 - \Gamma_2|_{1-var}.
    \end{equation*}
\end{proposition}

 Additional analytic and geometric properties of the signature, along with applications to rough paths is found in~\cite{friz_multidimensional_2010}.

\section{Applications to time series analysis}

The signature provides a faithful embedding of bounded variation paths into the formal power series algebra of tensors. By considering the truncated signature at some level $L \in \N$,
\begin{equation*}
    S_L(\Gamma) = 1+ \sum_{k=1}^L \sum_{I = (i_1, \ldots, i_k)} S^I(\Gamma) X_{i_1} \otimes \ldots \otimes X_{i_k},
\end{equation*}
we obtain a finite feature set $\{S^I(\Gamma)\}_{|I| \leq L}$ for a multi-dimensional time series, whose length does not depend on the length of the time series. One may draw parallels between the signature representation of a path and various series representations of functions such as Taylor series or Fourier series. However, there are two important differences:
\begin{enumerate}
    \item The set of Taylor series and Fourier series coefficients are linearly independent functionals, and provide a minimal set of features to describe functions. However, as described in the previous section, the full collection of path signatures $S(\cdot)$ is not independent and includes redundant information, though there do exist bases for the signature such as the Lyndon basis~\cite{reizenstein_calculation_2017}.
    
    \item Series representations of functions is linear, whereas the path signature is highly nonlinear. On the one hand, nonlinearity of the signature may capture nontrivial, discriminatory aspects of paths with fewer features than a linear representation. However on the other hand, nonlinearity causes the inversion problem of finding a path with a given signature to be significantly more difficult. A general method for continuous paths is given in~\cite{lyons_inverting_2018}, and another method for piecewise linear paths is given in~\cite{lyons_hyperbolic_2017}. An algebraic-geometric approach to the problem was recently established in~\cite{amendola_varieties_2018}.
\end{enumerate}

The feature set obtained from the truncated signature has recently been used in a variety of machine learning classification problems. Early examples include applications to financial time series~\cite{gyurko_extracting_2013} and handwritten character recognition~\cite{yang_deepwriterid:_2015}. Other examples include classifying time series of self-reported mood scores to distinguish between bipolar and borderline personality disorders~\cite{arribas_signature-based_2017}, and classifying time series of different brain region volumes to detect diagnosis of Alzheimers~\cite{moore_using_2018}. The surveys~\cite{chevyrev_primer_2016,lyons_rough_2014} further discuss these applications, along with different ways to transform the time series such that the path is better suited for signature analysis.

The path signature feature set has also been successful in situations where the data isn't naturally a path. This is the case in~\cite{chevyrev_persistence_2018} in which the path signature is used in conjunction with persistent homology to build a feature set for barcodes, a topological summary of a data set. Barcodes have no standard description as a vector of fixed dimension, and this method provides such a description, allowing techniques from topological data analysis to be used with standard machine learning algorithms. The proposed pipeline consists of the the following compositions
\begin{equation*}
    \textbf{Met} \xrightarrow{PH} \textbf{Bar} \xrightarrow{\iota   } BV(\R^N) \xrightarrow{S_L} \R\llangle X \rrangle.
\end{equation*}
The map $PH: \textbf{Met} \rightarrow \textbf{Bar}$ refers to the persistent homology functor, which assigns a barcode to the input data represented by a metric space (such as a point cloud in Euclidean space)~\cite{ghrist_barcodes_2008}. The barcode can then be transformed into a path in Euclidean space by the transformation $\iota$, and finally the truncated signature $S_L$ is computed. Several transformations $\iota$ from barcodes to paths are considered in the paper, and several are applied in this pipeline resulting in state-of-the-art performance on some standard classification benchmarks.

\medskip

These applications demonstrate the utility of using path signature terms for classification tasks. However, as posited in the opening discussion, the power of topological tools lies in their interpretability. Thus, we now turn our attention to the question of how path signatures provide encode human-understandable properties of multivariate time series. We begin with the notion of signed area and cyclicity, which is a way to study lead-lag relationships between time series in the absence of periodicity. This weak structure is difficult to capture with classical methods for time series analysis, which rely on the regularity of the parameterization to decompose the time series. To address this dificulty, Baryshnikov \cite{baryshnikov_cyclicity_2016} suggested the use of path signatures to characterize cyclicity. Next, we consider how the second level signature terms can be viewed as a measure of causality.


\subsection{Cyclicity and Lead-Lag Relationships} 
We begin by explicitly computing the first two levels of the path signature. Again, we consider a collection of $N$ simultaneous time series $\gamma_i:[0,1] \rightarrow \R$, viewed as a path $\Gamma \in P\R^N$. By definition, we can compute
\begin{align*}
    S^i(\Gamma) &= \int_0^1 \gamma'_i(t) dt = \gamma_i(1) - \gamma_i(0), \\
    S^{i,j}(\Gamma) & = \int_0^1 S^i(\Gamma)(t) \gamma'_j(t) dt = \int_0^1 (\gamma_i(t) - \gamma_i(0)) \gamma'_j(t) dt.
\end{align*}
The second level signature terms of a path in $\R^2$ are shown as the shaded areas in the following figure, where solid blue represents positive area, and hatched red represents negative area.

\begin{figure}[!htbp]
\centering
	\includegraphics[width=0.9\textwidth]{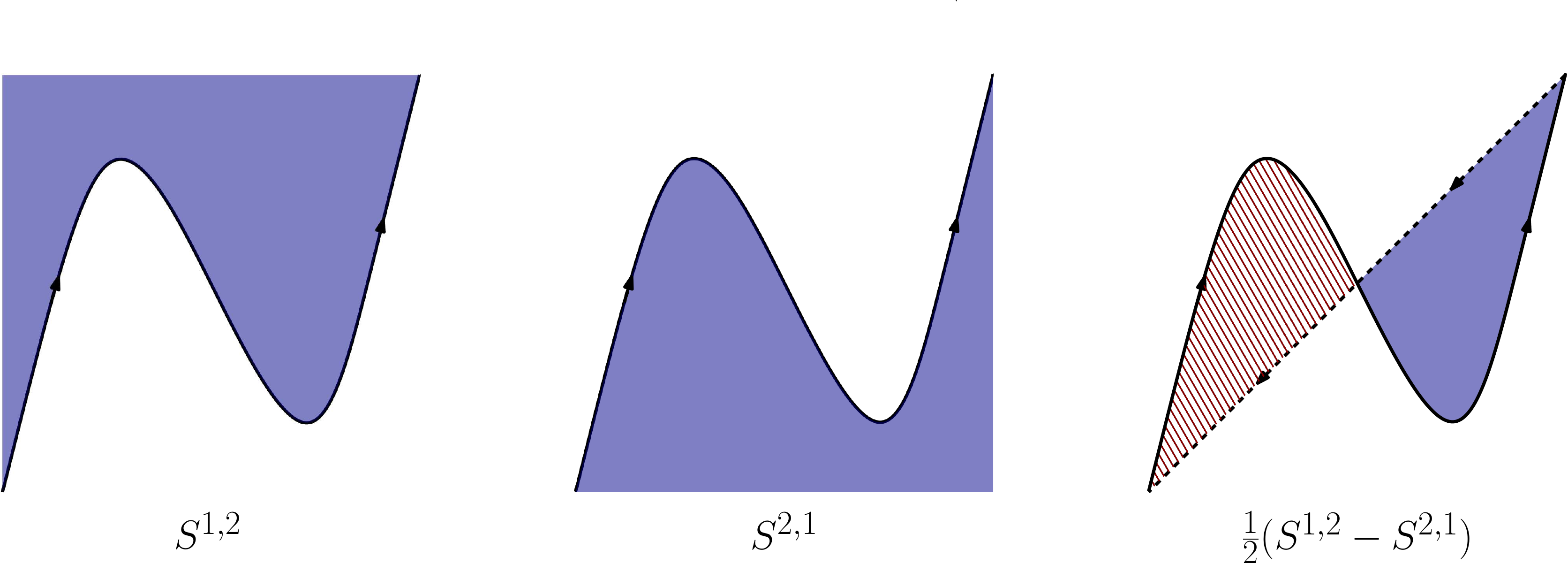}
\end{figure}

\noindent The third panel suggests that the linear combination $\frac12 (S^{i,j}(\Gamma) - S^{j,i}(\Gamma))$ encodes some information intrinsic to the path $\Gamma$. 

\begin{definition}
    Let $\alpha: [0,1] \rightarrow \R^2$ be a continuous closed curve defined by $\alpha(t) = (\alpha_1(t),\alpha_2(t))$ and $x = (x_1,x_2) \in \R^2 \backslash \textrm{im}(\alpha)$. We can rewrite $\alpha(t)$ in terms of polar coordinates $\alpha(t) = \left(r_{\alpha,x}(t), \theta_{\alpha,x}(t)\right)$ centered at $x$ where
    \begin{align*}
        r_{\alpha,x}(t) = |\alpha(t) - x|, \quad \theta_{\alpha, x}(0) = \tan^{-1}\left(\frac{\alpha_2(0) - x_2}{\alpha_1(0) - x_1}\right),
    \end{align*}
    and $\theta_{\alpha,x}(t)$ is defined via continuity. The \textit{winding number} of $\alpha$ with respect to $x$ is
    \begin{align*}
        \eta(\alpha,x) = \frac{\theta_{\alpha,x}(1) - \theta_{\alpha,x}(0)}{2\pi}.
    \end{align*}
\end{definition}

\begin{proposition}
    Suppose $\Gamma \in P\R^N$, and let $\tilde{\Gamma} = (\tilde{\gamma}_1, \ldots, \tilde{\gamma}_N)$ be the concatenation of $\Gamma$ with a linear path connecting $\Gamma(1)$ to $\Gamma(0)$. In addition, let $\tilde{\Gamma}_{i,j} = (\tilde{\gamma}_i(t), \tilde{\gamma}_j(t))$. Then
    \begin{equation*}
        A^{i,j}(\Gamma) \coloneqq \frac12\left( S^{i,j}(\Gamma) - S^{j,i}(\Gamma)\right) = \int_{\R^2} \eta(\tilde{\Gamma}_{i,j},x) \dd x
    \end{equation*}
    which is called the \textit{signed area}.
\end{proposition}

\begin{proof}
    We begin by assuming $\Gamma(0) = 0$ by translation invariance. Next, we show that $A^{i,j}(\Gamma) = A^{i,j}(\tilde{\Gamma})$. In the time interval $t\in[1/2,1]$, the components of the path $\tilde{\Gamma}$ can be written as
    \begin{equation*}
        \gamma_i(t) = m_i t + b_i
    \end{equation*}
    where $m_i=-b_i$ since the path must end at $\gamma_i(1) = 0$. Then, we have
    \begin{align*}
        A^{i,j}(\tilde{\Gamma}) & = \int_0^1 \gamma_i(t) \gamma_j'(t) - \gamma_j(t) \gamma'_i(t) dt \\
        & = A^{i,j}(\Gamma) + \int_{1/2}^1 (m_i t + b_i)m_j - (m_j t + b_j) m_i dt \\
        & = A^{i,j}(\Gamma).
    \end{align*}
    
    \noindent Now, suppose Finally, by applying Stokes' theorem, we get
    \begin{align*}
        A^{i,j}(\tilde{\Gamma}) &= \oint_{\tilde{\Gamma}} x_i dx_j - x_j dx_i \\
        & = \int_{\R^2} \eta(\tilde{\Gamma}_{i,j}, x) \dd x.
    \end{align*}

\end{proof}

In the third figure above, blue corresponds to a winding number of $1$ whereas red corresponds to a winding number of $-1$, resulting in the same interpretation as the formula. 
More generally, it was shown in~\cite{boedihardjo_uniqueness_2014} that all moments of the winding number of the curve $\tilde{\Gamma} - \tilde{\Gamma}(0)$ can be computed by linear combinations of signature terms of $\Gamma$, and conversely that the first four terms of $\log S(\Gamma)$ can be expressed using only the function $\eta(\tilde{\Gamma}-\tilde{\Gamma}(0),x)$.

\medskip

 The appearance of the winding number suggests that path signatures should be useful in studying periodic time series. However,  reparamerization-invariance  means that the signature naturally captures the broader and increasingly important class of \emph{cyclic} time series. Cyclic time series are those which can be factored through the circle
\begin{align*}
    \Gamma : [0,1] \xrightarrow{\phi} S^1 \xrightarrow{f} \R^N
\end{align*}
where $\phi$ is an orientation-preserving parametrization of the process. 
Cyclic phenomena arise naturally in a plethora of fields. Some simple examples include physiological processes such as breathing, sleep, the cardiac cycle, and neuronal firing; ecological processes such as the carbon cycle; and control processes involving feedback loops. Despite their repetitive nature, very rarely are such processes truly periodic, or even quasi-periodic, except to a coarse approximation. 


 One question of interest when studying cyclic processes is whether there exists a \emph{lead-lag relationship} between two or more signals; such a relationship may indicate causality, or simply provide a predictive signal. Consider the two pairs of time series $\Gamma^a = (\gamma^a_1, \gamma^a_2)$ and $\Gamma^b = (\gamma^b_1,\gamma^b_2)$,  shown on the left in the following figure. These two time series are chosen such that $\Gamma^b$ is simply a reparametrization of $\Gamma^a$, so there exists an orientation-preserving $\phi:[0,1]\rightarrow[0,1]$ such that $\Gamma^b = \Gamma^a \circ \phi$.

\begin{figure}[!htbp]   
\centering
	\includegraphics[width=\textwidth]{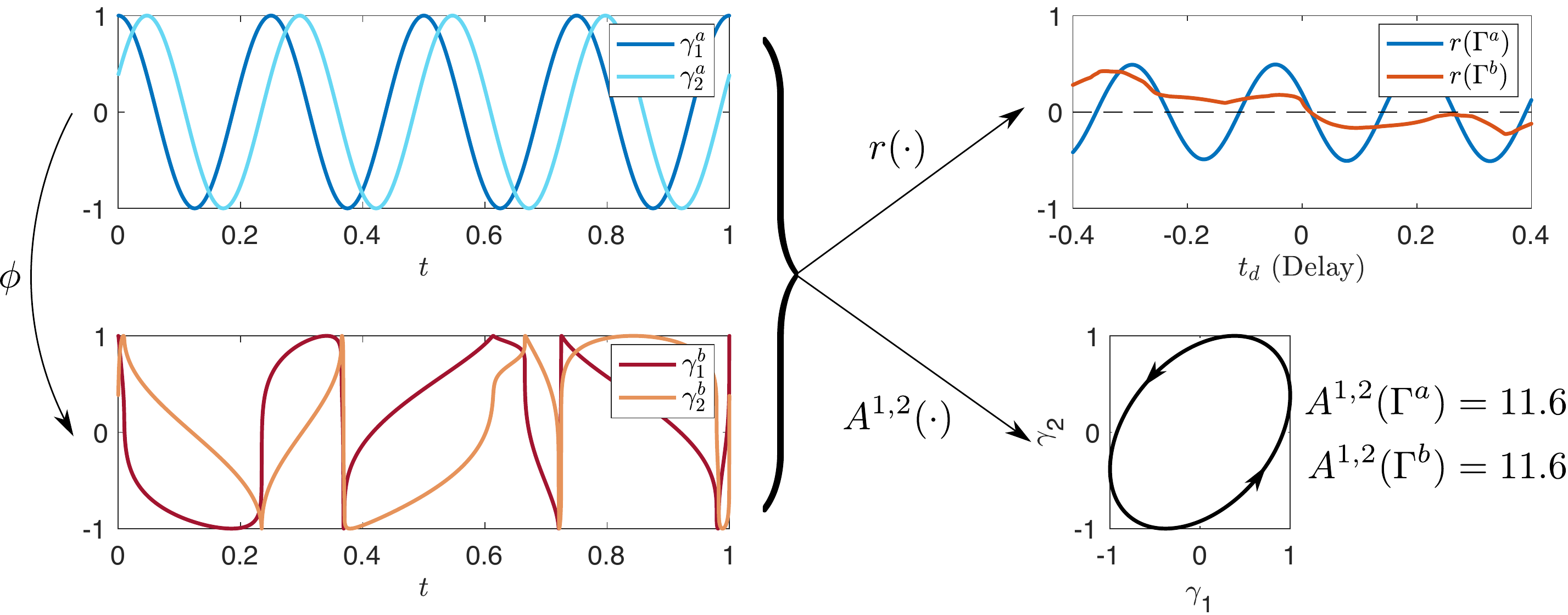}
\end{figure}

Perhaps the most common method for detecting lead-lag relationships in time series $\Gamma:[0,T] \rightarrow \R^2$ is the unbiased cross-correlation, defined by
\begin{equation*}
    r(\Gamma)(t_d) = \frac{1}{T-t_d} \int_0^T \gamma_1(t) \gamma_2(t-t_d) dt,
\end{equation*}
where $T$ is the total length of the time series and $\Gamma(t) = 0$ when $t \notin [0,T]$. The unbiased cross correlation of both sets of time series are shown on the top right. The cross correlation of $\Gamma^a$ has a clear periodic structure of its own, suggesting that the presence of a cyclic process in which one signal leads the other. The distance between maxima provides an estimate of the period of the two signals, and the phase-shift an estimate of the time-delay between $\gamma_1^a$ and $\gamma_2^a$. However, the cross correlation of $\Gamma^b$ is irregular, and though it attains a large value near $t_d = -0.4$, this is clearly not the primary scale on which the system is demonstrating cyclic behavior -- indeed, a constant scale doesn't exist.

However, since $\Gamma^b$ is a reparametrization of $\Gamma^a$, they will have the same signed area $A^{1,2}(\Gamma^a) = A^{1,2}(\Gamma^b)$. Indeed, the curve traced out by $(\gamma_1,\gamma_2)$, shown in the bottom right, which winds around counter-clockwise 4 times, indicating the four ``events" in each time series. The positive signed area suggests a lead-lag relationship for both sets of time series; this equivalence arises because the path signature depends only on ordered, simultaneous measurements, rather than the time between measurements.

In general, we can apply such an analysis to multidimensional time series by calculating the signed area between every pair of time series. In the context of sampled data, this computation boils down to the dot product of vectors, so is computationally feasible even for large systems, and the additivity of the integrals over partitions of domains means the measure can easily be implemented for streaming data.

\begin{definition}
    Let $\Gamma \in P\R^N$ represent $N$ simultaneous time series. The \textit{lead matrix} of $\Gamma$ is an $N \times N$ skew-symmetric matrix with entries
    \begin{align*}
        (A)_{i,j} = A^{i,j}(\Gamma).
    \end{align*}
\end{definition}

\begin{figure}[!htbp]   
\centering
	\includegraphics[width=0.85\textwidth]{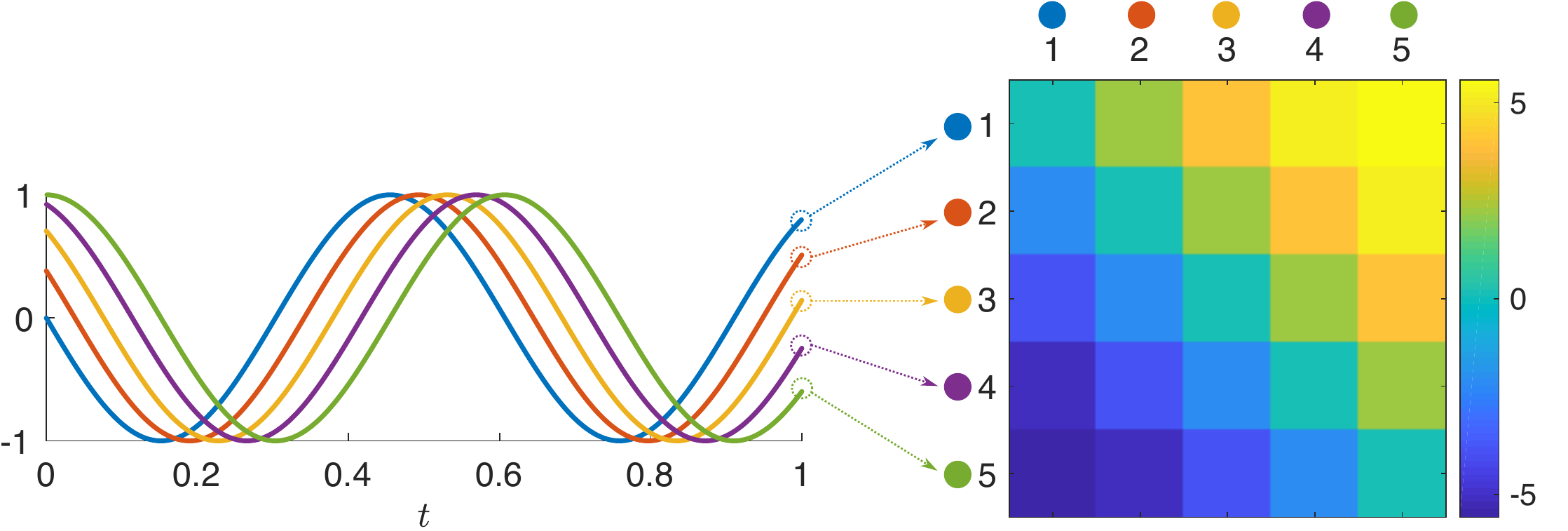}
\end{figure}

The matrix characterizes pairwise lead-lag behavior among a family of simultaneous time series. This method has been applied to the study of fMRI data, distinguishing between patients with tinnitus and those with normal hearing~\cite{zimmerman_dissociating_2018}. The skew-symmetric nature of this matrix lends itself to analogies with covariance matrices, however whereas the covariance matrix measures undirected and temporally independent relationships between variables, the lead matrix measures temporally directed relationships between variables.

\medskip

Of course, computing the signed area of the entire time series will only provide sensible lead-lag information if this behavior persists throughout the entire time interval. In many scenarios, this is not the case. For example, in gene regulatory networks there are cycles of activity initiated by irregular, external chemical signals. Different signals may induce different cycles of behavior, which may even have inverse lead-lag relationships, so integration across the entire time domain will provide negligible signature. Similarly, in an experimental environment we may perturb a system, necessarily leading to non-stationarity in the observed behavior, in which case the interesting signal would be the change in relationships acorss different epochs. Such controlled perturbations are, in particular, necessary for rigorous causality inference.

\begin{figure}[!htbp]   
\centering
	\includegraphics[width=\textwidth]{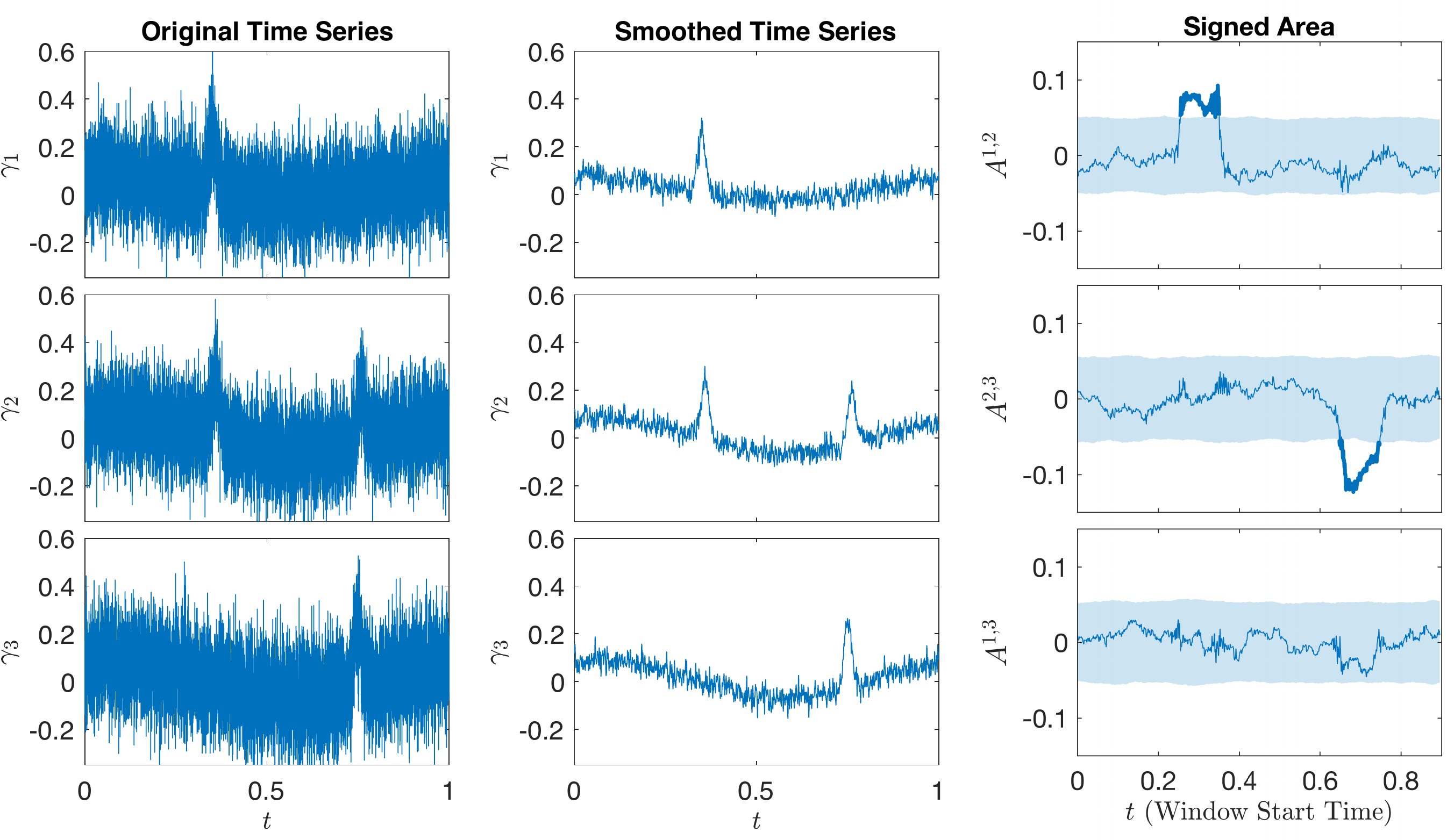}
\end{figure}

For example, consider the synthetic time series $\Gamma = (\gamma_1, \gamma_2, \gamma_3)$, as shown on the left column of the figure. We wish to detect whether or not there exist any lead-lag cycles that occur on a time scale that is small compared to the entire interval of the time series. Thus we perform signed area computations along a sliding window of the time series. We begin by convolving the time series with a narrow Gaussian as a smoothing preprocessing step to reduce noise. Next, we compute the three signed areas $A^{1,2}$, $A^{2,3}$ and $A^{1,3}$ along a sliding window of length $t=0.1$. 

To test statistical significance, we use a time shuffled null model, created by randomly permuting the elements of the time series within each component, and performing the same analysis (smoothing and sliding window signed area) on the shuffled time series. This is repeated 1000 times to generate a null distribution for the signed area curve of each component. The shaded portion of the signed area plot represents the $3\sigma$ confidence intervals in the third column panels.

While formal analysis of the probabilities requires a model of the underlying time series, we can empirically infer that a lead-lag relationship exists if the signed area is outside the confidence interval consecutively for a long sequence of consecutive time points. Thus, we likely we have an event with positive $A^{1,2}$, in which $\gamma_1$ leads $\gamma_2$, and also an event with negative $A^{2,3}$, in which $\gamma_3$ leads $\gamma_2$. 

This example demonstrates how the path signature may be used to detect lead-lag relationships in a model-free setting. The generality of the path signature can be exploited in other ways, and we describe a different interpretation of the second level signatures in terms of causality in the next section.

\subsection{Causality Analysis}

One of the fundamental steps in understanding the function of complex systems is the identification of causal relationships. However, empirically identifying such relationships is challenging, particularly when controlled experiments are difficult or expensive to perform. Three of the most common of approaches to causal inference are structural equation modelling, Granger causality, and convergent cross mapping. Like most approaches, these suffer from stringent assumptions that may not hold in empirical data. In order to understand these limitations, we first outline these methods, then describe how the second level signature terms can be applied as an assumption-free measurement of \textit{potential influences} in observational data, and explore some examples of their use.

We follow our previous notation and let $\Gamma(t) = (\gamma_1(t), \ldots, \gamma_N(t))$ denote a collection of $N$ simultaneous time series. In the following examples, we consider whether $\gamma_1(t)$ causally effects $\gamma_2(t)$; the rest of the time series should be interpreted as measured external factors. 

\medskip

Structural equation modelling (SEM)~\cite{wright_correlation_1921,haavelmo_statistical_1943} was one of the earliest developments in causal inference. It has more recently been recast into a formal framework by Pearl~\cite{pearl_causality:_2009} in which causal relationships can be determined. The fundamental operating principle of SEM is that causal assumptions are codified as hypotheses in the form of a directed graph, called a \textit{causal diagram}. The nodes represent all variables of interest, and directed edges represent possible causal influences. Note that the crucial information in such a diagram is the absence of edges.

\begin{figure}[!htbp]
\centering
	\includegraphics[width=0.2\textwidth]{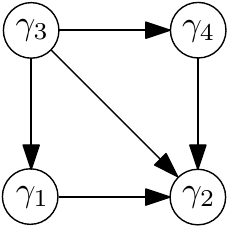}
\end{figure}

Given this causal diagram, the structural equation most commonly used in practice for time series assumes linearity, Gaussian errors and stationarity~\cite{chen_vector_2011, kim_unified_2007}. It can be viewed as a combination of linear SEM and a vector autoregressive (VAR) model,
\begin{equation*}
    \Gamma(t) = \sum_{i=0}^n \beta_i \Gamma(t-i) + U(t)
\end{equation*}
where $\beta_i$ is a matrix of effect sizes for a given time lag, and $U$ is a vector of random Gaussian variables which represents error. The causal assumptions are encoded in $\beta_i$, which has a zero entry for every directed edge that is omitted from the causal diagram. The goal is then to estimate the parameters $\beta_i$ based on empirical data to determine whether or not causal influences exist.

\medskip

Another measure of causality in common use is Granger causality~\cite{granger_investigating_1969}, which explicitly accounts for the temporal nature of causality, and is often used with time series data. It operates based on two main principles.
\begin{enumerate}
    \item (Temporal precedence) The effect does not precede its cause in time.
    \item (Separability) The causal series contains unique information about the effected series that is otherwise not available.
\end{enumerate}

Let $A \independent B \;|\; C$ denote that $A$ and $B$ are independent given $C$ and let $X^t = \{X(s) \, | \, s \leq t\}$ denote the history of $X(t)$ up to time $t$. 

\begin{definition}
    The process $\gamma_1(t)$ is \textit{Granger non-causal} for the series $\gamma_2(t)$ with respect to $\Gamma = (\gamma_1(t), \gamma_2(t), \gamma_3(t))$ if
    \begin{align*}
        \gamma_2(t+1) \independent \gamma_1^t \; | \; \gamma_2^t, \gamma_3^t
    \end{align*}
    for all $t \in \Z$; otherwise $\gamma_1(t)$ \textit{Granger causes} $\gamma_2(t)$ with respect to $\Gamma$.
\end{definition}
The idea behind this definition is that $\gamma_1$ does not causally influence $\gamma_2$ if future values of $\gamma_2$ are independent to all past values of $\gamma_1$, conditioned on past values of $\gamma_2$ and any external factors $\gamma_3$.



A measure of Granger causality is determined by a comparison of predictive power~\cite{bressler_wienergranger_2011}. Let $\Gamma = (\gamma_1,\gamma_2,\gamma_3)$ and $\widetilde{\Gamma} = (\gamma_2,\gamma_3)$, and we assume that these time series are modeled by a VAR process. To test the criteria of independence in Granger causality, we fit two VAR models
\begin{align}
    \Gamma(t) &= \sum_{i=1}^n A_i \Gamma(t) + U(t), \\
    \widetilde{\Gamma}(t) &= \sum_{i=1}^n \widetilde{A}_i \widetilde{\Gamma}(t) + \widetilde{U}(t).
\end{align}
Prediction accuracy of either model is determined by the variance of the residual $\textrm{var}(U(t))$. Thus, the empirical notion of Granger causality is defined by
\begin{equation*}
    C_{\gamma_1 \rightarrow \gamma_2} = \ln \frac{\textrm{var}(\widetilde{U}(t))}{\textrm{var}(U(t))}.
\end{equation*}

The separability assumption is untrue in many situations. Prominent examples are deterministic dynamical systems with coupling between variables such as in a feedback loop. This is clear from Taken's theorem.

\begin{theorem}[\cite{takens_detecting_1981}]
    Let $M$ be a compact manifold of dimension $m$. For pairs $(\phi, y)$, where $\psi: M \rightarrow M$ is a diffeomorphism and the observation function $y: M \rightarrow \R$ is smooth, it is a generic property that the map $\Psi: M \rightarrow \R^{2m+1}$ defined by 
    \begin{equation*}
    \Psi(x) = \left(y(x), y(\psi(x)), y(\psi^2(x)), \ldots, y(\psi^{2m}(x))\right)
    \end{equation*}
    is an embedding.
\end{theorem}

Here we treat $M$ as an invariant manifold of a dynamical system evolving according to a vector field $V$, and the diffeomorphism $\psi$ corresponds to the flow of $V$ with respect to negative time $-\tau$. The observation function is usually taken to be a projection map $\pi_i$ on to the $X_i$ coordinate. In this context, Taken's theorem states that the manifold $M$ is diffeomorphic to reconstructions via the delay embedding $\Psi_i$ using any of the projection maps $\pi_i$, assuming they are generic. Thus, if two variables $X_i$ and $X_j$ are coupled in the dynamical system, then information about the state of one variable $X_i$ exists in the history of another $X_j$.

\medskip

The final approach to causal inference that we describe takes advantage of this property of dynamical systems. The method of convergent cross mapping (CCM) was developed by Sugihara ~\cite{sugihara_detecting_2012} and later placed in a rigorous mathematical framework~\cite{cummins_efficacy_2015}. The motivation behind CCM is to understand the causal structure of an $N$ dimensional time series $\Gamma(t)$ which is a trajectory of an underlying deterministic dynamical system
\begin{align*}
    \gamma_i'(t) = V_i(\gamma_1,\ldots, \gamma_N).
\end{align*}
A component $\gamma_1(t)$ causally influences component $\gamma_2(t)$ if the $\gamma_2$ component of the vector field, $V_2(X)$ has a nontrivial dependence on $\gamma_1$. The idea is that if such a nontrivial dependence exists, then one can predict the states in $M_1$ based on the information in $M_2$. Prediction accuracy should increase as we include more time points in $\Gamma(t)$, and the convergence of prediction accuracy is used as the indication of causal influence. 

\medskip

The three methods of causal inference surveyed here were established based on different notions of causality, and are thus applicable in different scenarios. However, the practical implementations of SEM and GC depend on strong assumptions such as linearity, stationarity and Gaussian noise, which often do not hold for empirical data. Moreover, SEM requires a priori knowledge about the underlying process which may not be well established for complex data sets. CCM moves beyond linear and stationary assumptions to study complex nonlinear systems, but still depends on a dynamical systems model.

We propose the path signature as a model-free measure of causality, in which our only assumption is that of temporal precedence of causal effects. Namely, we wish to detect the \textit{observed influence} between the various components in our time series. We do not claim that observed influences are truly causal. Omitted external factors may confound observed variables, and various true causal pathways may result in spurious influences.

This approach is motivated by the equation for the second level of the signature, in the case where $\gamma_i(0) = 0$ for all components $i$,
\begin{align*}
    S^{i,j}(\Gamma) = \int_0^1 \gamma_i(t) \gamma_j'(t) dt.
\end{align*}
Here the term $\gamma_i(t)$ should be thought of as the distance from the mean of the path component. In practice, this is done by translating each component of the path such that it has mean $0$, normalizing the time series to have maximum value $1$ (either separately or as a group, depending on the intended application), and appending $\gamma_i(0) = 0$ at the beginning of each time series.  

With this context,  the integrand can be viewed as a measure of how the magnitude of $\gamma_i(t)$ \textit{influences} the change in $\gamma_j'(t)$. By integrating over the entire path, we obtain an aggregate measure of the influence of $\gamma_i$ on the change in $\gamma_j$ over the given time interval. As such, the second order signatures provide a measure of potential observed influence, indicating possible causal relationships between variables using only observations of time series, without any prior assumptions. Of course, this method will not be able to distinguish between true and spurious causal relations (due to confounders, for example). However, such caveats would necessarily apply to any system in the absence of a model; thus, in addition to providing a coarse measure of causality, one can view this method as a preprocessing step for the model-based methods described above.

\medskip

We close with a final example, considering the case that the system is known or suspected to be non-stationary. In this setting, a global measure of influence is inappropriate, as we are often interested in the change in such structure when the system changes modes. Fortunately, it is straightforward to modify the signature measure to detect temporally localized influences. This is done by studying the derivative of the signature, which is simply given by the integrand
\begin{equation*}
    (S^{i,j})'(\Gamma)(t) = \gamma_i(t)\gamma_j'(t).
\end{equation*}
Geometrically, this is the instantaneous area of the arc at the origin of the $(\gamma_i, \gamma_j)$-plane swept out by the pair of time series. If this measure has large magnitude on an interval, it suggests suggests that one of the series is strongly influencing the other during that epoch.

We demonstrate this method using a familiar example of dynamics which exhibit mode-switching.  Consider the time series $\Gamma(t) = (\gamma_1(t), \gamma_2(t), \gamma_3(t))$, which represents a portion of a discretized solution to the Lorenz equations
    \begin{align*}
        \gamma_1'(t) &= \sigma(\gamma_2(t) - \gamma_1(t)), \\
        \gamma_2'(t) &= \gamma_1(t)(\rho - \gamma_3(t)) - \gamma_2(t), \\
        \gamma_3'(t) & = \gamma_1(t)\gamma_2(t) - \beta\gamma_3(t),
    \end{align*}
    where we have taken the parameters $\sigma = 10$, $\rho = 28$, and $\beta = 8/3$. The equations are solved using the built-in \texttt{ode45} function in MATLAB. For preprocessing, each component has been translated so that it has mean $0$, and an additional point has been appended to the beginning of the time series so that it starts at the origin. Each component is individually normalized so that $\sup(\gamma_i(t)) - \inf(\gamma_i(t)) = 1$.

    In the following figure, all six second-level signature derivative terms are shown on the left, with plots of the path projected onto the corresponding plane. Note that the time axis has arbitrary units due to reparametrization invariance. As with the previous example, we use a time shuffled null model in which the same analysis is performed on the shuffled time series. The null distribution is generated by repeating this procedure 1000 times. The shaded portion of the signature plots correspond to the $3\sigma$ confidence intervals.
    
    The upper and lower bounds of the confidence intervals are outlined with red and green lines respectively. The time points at which the the signature derivative is either above or below the confidence interval are considered significant, and are respectively colored red or green in the plot on the right.
    
    We observe the expected result in the first row: the signature derivative picks out sections of the plot in which $\gamma_1$ is positive (negative) and $\gamma_2$ is increasing (decreasing). The opposite trend of sections in which $\gamma_1$ is positive (negative) and $\gamma_2$ is decreasing (increasing) is seen in the green time points. 

    \begin{figure}[!htbp]
    \centering
	    \includegraphics[width=0.9\textwidth]{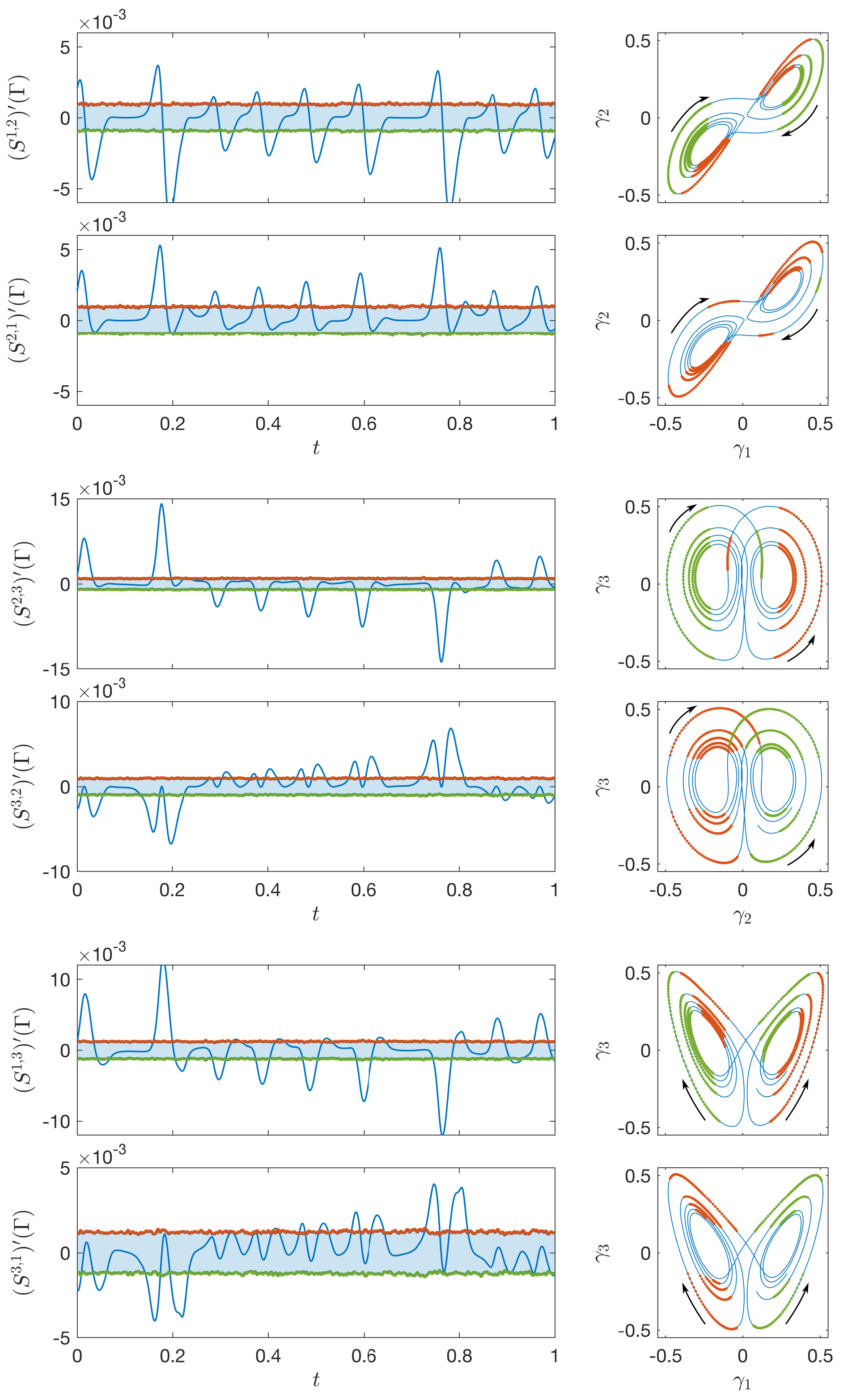}
    \end{figure}

    \bigskip
    
    
\section{Generalizations and Outlook}
\label{sec:outlook}

We have seen that path signatures provide a natural feature set for studying multivariate time series. In addition, we have discussed ways to view the second level signature terms in order to understand the path signature in an interpretable manner.
In this section, we outline two directions for generalizations of these ideas to more complex settings, which will be further discussed in forthcoming work by the authors.\medskip

The first direction is to consider the full Chen cochain model $Chen(P\R^N)$, alluded to in Section~\ref{sec:iterated} and further discussed in Appendix~\ref{sec:pathspace}, which is a subcomplex of the de Rham complex of differential forms on $P\R^N$. The iterated integrals described thus far are the $0$-forms in this cochain model, and we have seen that these cochains describe properties of individual points of $P\R^N$. 

Integration of the $1$-forms of $Chen(P\R^N)$ along paths in $P\R^N$, interpreted as parametrized families of time series, provides information about such a family. To draw an analogy, consider the case of differential forms on $\R^N$. The $0$-forms are simply functions, which provide information about individual points in $\R^N$, while integration of $1$-forms provide information about paths in $\R^N$. For example, integration of $\dd x_i$ along a path tells us the displacement in the $x_i$ coordinate. \medskip

The simplest example of a $1$-form in $Chen(P\R^N)$ is generated by a single $2$-form on $\R^N$. We follow the construction in Definition~\ref{def:iiform} to obtain our desired $1$-form. 

Suppose $\omega = \dd x_i \wedge \dd x_j$, and suppose $\overline{\alpha}: I \rightarrow P\R^N$ is a family of paths. Associated to such a family is the map $\alpha : I \times I \rightarrow \R^N$, defined by $\alpha(s,t) = \overline{\alpha}(s)(t)$. The pullback of $\omega$ with respect to $\alpha$ is
\begin{equation*}
    (\alpha)^*(\omega) = \left( \frac{\partial \alpha_i}{\partial s}\frac{\partial \alpha_j}{\partial t} - \frac{\partial \alpha_i}{\partial t} \frac{\partial \alpha_j}{\partial s}\right) \dd s \wedge \dd t.
\end{equation*}
The $1$-form in $Chen(P\R^N)$ with respect to $\omega$, viewed under the plot $\overline{\alpha}$ is defined to be
\begin{equation*}
    \left( \int \omega\right)_{\overline{\alpha}} = \left( \int_0^1 \frac{\partial \alpha_i}{\partial s}\frac{\partial \alpha_j}{\partial t} - \frac{\partial \alpha_i}{\partial t} \frac{\partial \alpha_j}{\partial s} \dd t \right) \dd s.
\end{equation*}
We can think of this expression as the pullback of the $1$-form $\int \omega$ along $\overline{\alpha}$. Thus, integrating over the family of paths corresponds to integrating over $s$, and we obtain
\begin{equation*}
    \int_{\overline{\alpha}} \left( \int \omega\right) = \int_0^1 \int_0^1 \frac{\partial \alpha_i}{\partial s}\frac{\partial \alpha_j}{\partial t} - \frac{\partial \alpha_i}{\partial t} \frac{\partial \alpha_j}{\partial s} \dd t \, \dd s.
\end{equation*}
Note that the integrand is the determinant of the Jacobian of $\alpha_{i,j} = (\alpha_i, \alpha_j)$. Therefore, integration of $\int \omega$ along a family of paths yields the area of the region $\alpha_{i,j}(I^2)$, as shown in the figure.

\begin{figure}[!htbp]
\centering
    \includegraphics[width=0.65\textwidth]{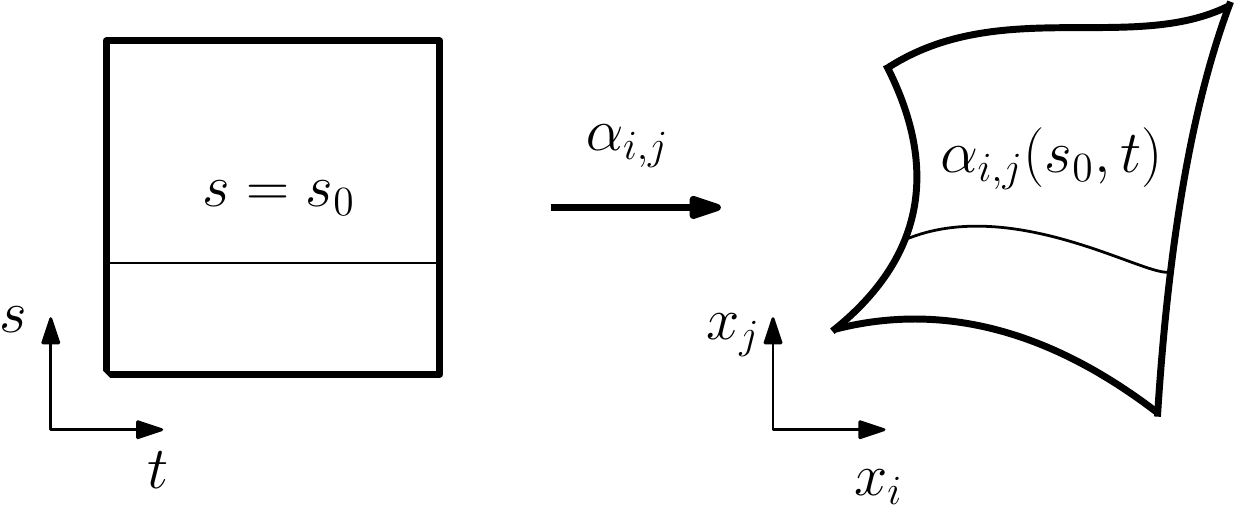}
\end{figure}

Although the information in $\int\omega$ may seem elementary, this example produces the simplest $1$-form by Chen's construction, analogous to the first level signature terms $S^i$. The idea is to mimic the construction of the path signature, and construct iterated integral forms out of a $2$-form $\omega \in A^2_{dR}(\R^N)$ and several $1$-forms $\omega_i \in A^1_{dR}(\R^N)$ in different permutations to acquire more sophisticated properties of these families of paths. In fact, by considering the $p$-forms in $Chen(P\R^N)$, we can study multiparameter families of paths in a similar manner. \medskip

The second direction is to consider spaces more general than $P\R^N$. Namely, we can think of the path space as the mapping space $P\R^N = Map(I, \R^N)$, and consider iterated integral cochain models for more general mapping spaces. In fact, Chen's definition of the path signature was not restricted to paths on $\R^N$, but rather paths on differentiable manifolds $M$. The definition of the path signature is the same as Definition~\ref{def:ii}, except we replace the standard $1$-forms with a collection of forms $\omega_1, \ldots, \omega_m \in A^1_{dR}(M)$. Most of the algebraic properties from Section~\ref{sec:iterated} still hold for path signatures in $PM$. In the following theorem $C^r$ denotes $r$-times continuously differentiable.

\begin{theorem}[\cite{chen_integration_1958}]
Let $M$ be a $C^r$ manifold with $r \geq 2$, and suppose $\omega_1, \ldots, \omega_m \in A^1_{dR}(M)$ such that they span the cotangent bundle $T^*M$ at every point. Then if $\Gamma_1, \Gamma_2 \in PM$ are irreducible piecewise-regular continuous paths such that $\Gamma_1(0) = \Gamma_2(0)$ and $S(\Gamma_1) = S(\Gamma_2)$, then $\Gamma_1$ is a reparametrization of $\Gamma_2$.
\end{theorem}

This is the analogous statement of Theorem~\ref{thm:injective} but for $PM$ rather than $P\R^N$. This theorem states that the path signature for manifolds is still a faithful representation of paths. Thus, it provides a complete reparametrization-invariant feature set for multivariate time series that naturally lie on a manifold. For example, time series of phases of a collection of oscillators, would be a path on a toroidal manifold. Another example may be time series of states of some dynamical system, which may be a trajectory on an invariant manifold. \medskip

We can generalize further and consider mapping spaces $Map(Y,M)$, where $Y$ is a topological space such that $\textrm{conn}(M) \geq \dim(Y)$. In this case, there exists a generalized iterated integral cochain model for the mapping space, which is developed in~\cite{patras_cochain_2003, ginot_chen_2010}. This setting would allow for the study of data which is naturally modeled by elements of such mapping spaces. Possible examples include vector fields over an embedded manifold $M \subset \R^N$, which can be modelled by the mapping space $Map(M, \R^N)$.


\section*{Acknowledgements}
D.L. is supported by the Office of the Assistant Secretary of Defense Research \& Engineering through ONR N00014-16-1-2010, and the Natural Sciences and Engineering Research Council of Canada (NSERC) PGS-D3.

\appendix
\section{Path Space Cochains}
\label{sec:pathspace}

Chen's formulation of a cochain model begins by defining a de Rham-type cochain complex $A_{dR}$ on a general class of spaces called \textit{differentiable spaces}, generalizing the usual differential forms defined on manifolds. Path spaces are examples of differentiable spaces, and thus are associated with such a de Rham cochain complex. By defining iterated integrals using higher-degree forms on $\R^N$, rather than the $1$-forms used in Definition~\ref{def:ii}, we obtain forms on $P\R^N$ rather than functions. Finally, he shows that the forms generated by iterated integrals form a subcomplex of $A_{dR}$, and is in fact quasi-isomorphic to $A_{dR}$. A detailed account of this construction is found in~\cite{chen_iterated_1977}, and a more modern treatment can be found in~\cite{felix_algebraic_2008}.

Smooth structures are defined on manifolds by using charts to exploit the well-defined notion of smoothness on Euclidean space. Charts can be viewed as probes into the local structure of a manifold. However, as homeomorphisms of some Euclidean space of fixed dimension, charts are a rather rigid way to view local structure as they are both maps into and out of a manifold. Differentiable spaces relax this homeomorphism condition, and only require its \textit{plots}, the differentiable space analog of a chart, to map \textit{into} the space. Baez and Hoffnung~\cite{baez_convenient_2011} further discuss these ideas, along with categorical properties of differentiable spaces.
\begin{definition}
    A \textit{differentiable space} is a set $X$ equipped with, for every Euclidean convex set $C \subseteq \R^n$ with nonempty interior and for any dimension $n$, a collection of functions $\phi: C \rightarrow X$ called \textit{plots}, satisfying the following:
    \begin{enumerate}
        \item (Closure under pullback) If $\phi$ is a plot and $f:C' \rightarrow C$ is a smooth, then $\phi f$ is a plot.
        \item (Open cover condition) Suppose the collection of convex sets $\{C_j\}$ form an open cover of the convex set $C$, with inclusions $i_j: C_j \xhookrightarrow{} C$. If $\phi i_j$ is a plot for all $j$, then $\phi$ is a plot.
        \item (Constant plots) Every map $f:\R^0 \rightarrow X$ is a plot.
    \end{enumerate}
\end{definition}

It is clear that any manifold is a differentiable space by taking all smooth maps $\phi: C \rightarrow M$ to be plots. We obtain a canoncial differentiable space structure on $PM$ by noting that, given any map $\overline{\alpha}: C \rightarrow PM$, there is an associated \textit{adjoint map} $\alpha: I \times C \rightarrow M$ defined by $\alpha(t, x) = \overline{\alpha}(x)(t)$. Consider the collection of all maps $\overline{\alpha}: C \rightarrow M$ for which the adjoint $\alpha$ is a smooth map, which clearly satisfies the first and third conditions. To obtain a collection of plots on $PM$, we additionally include all maps $\overline{\alpha}: C \rightarrow PM$ such that the hypothesis of the second condition is true. 

\begin{definition}
    A \textit{$p$-form} $\omega$ on a differentiable space $X$ is an assignment of a $p$-form $\omega_\phi$ on $C$ to each plot $\phi: C \rightarrow X$ such that if $ f: C' \rightarrow C$ is smooth, then $\omega_{\phi f} = f^* \omega_{\phi}$. The collection of $p$-forms on $X$ is denoted $A^p_{dR}(X)$, and the graded collection of all forms on $X$ is $A_{dR}(X)$.
\end{definition}

Linearity, the wedge product, and the exterior derivative are all defined plot-wise. Namely, given $\omega, \omega_1, \omega_2 \in A_{dR}(X)$, $\lambda\in \R$, and any plot $\phi: C \rightarrow X$, 
\begin{itemize}
    \item $(\omega_1 + \lambda \omega_2)_\phi = (\omega_1)_\phi + \lambda(\omega_2)_\phi$,
    \item $(\omega_1 \wedge \omega_2)_\phi = (\omega_1)_\phi \wedge (\omega_2)_\phi$, and
    \item $(d\omega)_\phi = d\omega_\phi$.
\end{itemize}

Therefore, $A_{dR}(X)$ has the structure of a commutative differential graded algebra, and we may define the de Rham cohomology
\begin{align*}
    H^*_{dR}(X) \coloneqq H^*(A_{dR}(X))
\end{align*}
of differentiable spaces.

From here forward, we will focus on the case of forms on $P\R^N$, for which there is a special, easily understood class of forms defined using iterated integrals. Much of what we explicitly construct can be lifted to paths in manifolds of interest,or to more general mapping spaces, and will be discussed in forthcoming work by the authors.

\begin{definition}
\label{def:iiform}
    Let $\omega_1, \ldots, \omega_k$ be forms on $\R^N$ with $\omega_i \in A^{q_i}_{dR}(\R^N)$. The iterated integral $\int \omega_1 \ldots \omega_k$ is a $\left( (q_1 + \ldots + q_k) - k \right)$-form on $P\R^N$ defined as follows.
    Let $\overline{\alpha}:C \rightarrow P\R^N$ be a plot with adjoint $\alpha:C \times I \rightarrow \R^N$. Decompose the pullback of $\omega_i$ along $\alpha$ on $C \times I$ as
    \begin{equation*}
        \alpha^*(\omega_i)(x,t) = \dd t \wedge \omega_i'(x,t) + \omega_i''(x,t)
    \end{equation*}
    where $\omega_i', \omega_i''$ are $q_i$-forms on $C \times I$ without a $\dd t$ term. Then, the iterated integral is defined as
    \begin{equation*}
        \left(\int \omega_1 \ldots \omega_k\right)_{\overline{\alpha}} = \int_{\Delta^k} \omega_1'(x,t_1) \wedge \ldots \wedge \omega_k'(x,t_k) \, \dd t_1 \ldots \dd t_k.
    \end{equation*}
\end{definition}

Consider the conceptual similarities between this definition, and the one given in Definition~\ref{def:ii}. In the language of our present formulation, $S^I(\Gamma)$, as given in Equation~\ref{eq:ii_function}, is the iterated integral where $\omega_l = \dd x_{i_l}$ viewed through the one-point plot $\overline{\alpha}_\Gamma: \{*\} \rightarrow P\R^N$ defined by $\overline{\alpha}_\Gamma(*) = \Gamma$.

\begin{definition}
    Let $Chen(P\R^N)$ be the sub-vector space of forms on $P\R^N$ generated by
    \begin{align*}
        \pi_0^*(\omega_0) \wedge \int \omega_1 \ldots \omega_k \wedge \pi_1^*(\omega_{k+1})
    \end{align*}
    where
    \begin{itemize}
        \item $\omega_i \in A_{dR}(\R^N)$, for $i = 0, \ldots, k+1$,
        \item $\int \omega_1 \ldots \omega_k$ is the iterated integral in the previous definition, and
        \item $\pi_0, \pi_1: P\R^N \rightarrow \R^N$ are the evaluation maps at $0$ and $1$ respectively.
    \end{itemize}
\end{definition}

\begin{theorem}[\cite{chen_iterated_1977}]
\label{thm:dgsa}
    The complex $Chen(P\R^N)$ is a differential graded subalgebra of $A_{dR}(P\R^N)$.
\end{theorem}

This theorem is proved by showing that the $Chen(P\R^N)$ is closed under the differential and the wedge product. As we will not make use of the details, we refer the reader to~\cite{chen_iterated_1977} for further discussion of the differential, noting only that the additional forms $\pi_0^*(\omega_0)$ and $\pi_{1}^*(\omega_{k+1})$ are required for closure. The wedge product structure is analogous to the shuffle product identity in Theorem~\ref{thm:shuffle}, and is proved in a similar manner. Note that the wedge product structure for $0$-cochains is exactly Theorem~\ref{thm:shuffle}.

Given $m$ forms $\omega_i \in A^{q_i}_{dR}(\R^N)$ and $\sigma$ a permutation of the set $[m]$, we denote by $\epsilon_{\sigma, (q_i)} \in \{-1, 1\}$ the sign such that
\begin{align*}
    \omega_1 \wedge \ldots \wedge \omega_m = \epsilon_{\sigma, (q_i)} \left(\omega_{\sigma(1)} \wedge \ldots \wedge \omega_{\sigma(m)}\right).
\end{align*}
As the notation suggests,  $\epsilon_{\sigma, (q_i)}$ depends on both the permutation and the ordered list of the degrees $(q_i)$. 

\begin{lemma}
\label{lem:shuffle}
    Let $\omega_i \in A^{q_i}_{dR}(R^N)$ for $i = 1, \ldots, k+l$. We have the following product formula:
    \begin{align}
        \int \omega_1 \ldots \omega_k \wedge \int \omega_{k+1} \ldots \omega_{k+l} = \sum_{\sigma \in Sh(k,l)} \epsilon_{\sigma,(q_i)} \int \omega_{\sigma(1)} \omega_{\sigma(2)} \ldots \omega_{\sigma(k+l)}.
    \end{align}
\end{lemma}


Theorem~\ref{thm:dgsa} and the following theorem show that the subcomplex of iterated integrals $Chen(P\R^N)$ is a cochain model for $P\R^N$.

\begin{theorem}
    The two commutative differential graded algebras, $A_{dR}(P\R^N)$ and $Chen(P\R^N)$, have the same minimal model as $\R^N$.
\end{theorem}

Returning our focus to iterated integrals as functions, we see that the $S^I$ are $0$-cochains in this model, constructed via pullback and integration. Indeed, consider the evaluation map $\ev_k : \Delta^k \times P\R^N \rightarrow (\R^N)^k$ defined by
\begin{align*}
    \ev_k((t_1, \ldots, t_k), \Gamma) \coloneqq \left(\Gamma(t_1), \ldots, \Gamma(t_k)\right).
\end{align*}
Then, $S^I$ is the image of $\otimes_{l=1}^k \dd x_{i_l}$ under the composition
\begin{align*}
    A^1_{dR}(\R^n)^{\otimes k} \xrightarrow{\ev_k^*} A^k_{dR}(\Delta^k \times P\R^N) \xrightarrow{\int_{\Delta^k}} Chen^0(P\R^N).
\end{align*}

\bibliographystyle{amsplain}
\bibliography{abel.bib}

\end{document}